\documentclass[11pt, reqno]{amsart}       
\usepackage{graphicx}

\usepackage{amsmath}

\usepackage{amssymb}
\usepackage{pifont}
\usepackage[colorlinks=true,
linkcolor=blue,
citecolor=blue,
urlcolor=blue]{hyperref}
\newtheorem{thm}{Theorem}
\newtheorem{lem}{Lemma}

\newtheorem{cor}{Corollary}
\theoremstyle{definition}
\newtheorem{defn}{Definition}
\newtheorem*{rem}{Remarks}

\newtheorem*{rem2}{Remark}
\newtheorem{example}{Example}
\newtheorem*{acknowledgement}{Acknowledgement}

\newcommand{\vertiii}[1]{{\left\vert\kern-0.25ex\left\vert\kern-0.25ex\left\vert
		#1 \right\vert\kern-0.25ex\right\vert\kern-0.25ex\right\vert}}

\def \lim   {\text {\rm lim}}

\begin{document}
	
	\title[]{Approximate recoverability and the quantum data processing inequality}
	
	\author{Saptak Bhattacharya}
	
	\address{Theoretical Statistics and Mathematics Unit\\
		Indian Statistical Institute\\
		New Delhi 110016\\
		India}
	\email{saptak21r@isid.ac.in}
	
	
	
	
	\begin{abstract}
		In this paper, we discuss the quantum data processing inequality and its refinements that are physically meaningful in the context of approximate recoverability. An important conjecture regarding this due to Seshadreesan et. al. in J. Phys. A: Math. Theor. 48 (2015) is disproved. We prove some inequalities capturing universal approximate recoverability with the Petz recovery map for the sandwiched quasi and R\'enyi relative entropies for the parameter $t=2$. We also obtain convexity theorems on some parametrized versions of the relative entropy and fidelity, which can be of independent interest.
	\end{abstract}
	\subjclass[2010]{ 94A17, 15A45}
	
	\keywords{log-convexity, relative entropy, fidelity, recoverability
	}
	\date{}
	\maketitle
	
	\section{Introduction}

	Let $A, B$ be $n\times n$ density matrices with $\textrm{supp}(A)\subset\textrm{supp}(B)$. The {\it relative entropy}, or the {\it Kullback-Liebler divergence} of $A$ with respect to $B$ is defined as
	\[D(A|B)=\textrm{tr}A(\ln A-\ln B).\label{e1}\tag{1}\] This measures the distinguishability of $A$ with respect to $B$. Let $\phi:M_n(\mathbb{C})\to M_k(\mathbb{C})$ be a quantum channel. This is a completely positive, trace preserving map.
	The data processing inequality, originally proved by Lindblad in \cite{gl} is a result of fundamental importance in quantum information theory. It states that \[D(\phi(A)|\phi(B))\leq D(A|B).\label{e2}\tag{2}\] Intuitively, this means that states become more indistinguishable when a quantum channel is applied.
	\medskip

	It was proved by Petz \cite{pr} that equality holds in \eqref{e2} if and only if there exists a channel $\mathcal{R}:M_k(\mathbb{C})\to M_n(\mathbb{C})$ depending only on $\phi$ and $B$ such that $\mathcal{R}(\phi(B))=B$ and $\mathcal{R}(\phi(A))=A$. Petz constructed such a channel explicitly, henceforth known as the {\it Petz recovery map} \cite{hp, pr}. This is defined as \[\begin{aligned}&\mathcal{R}(Y)&=B^{1/2}\phi^{*}\big(\phi(B)^{-1/2}Y\phi(B)^{-1/2}\big)B^{1/2}.\end{aligned}\label{e3}\tag{3}\] This means that there exists a channel depending only on $\phi$ and $B$, which recovers $B$ from $\phi(B)$ while also recovering all such $A$ for which there is no loss of distinguishability information with respect to $B$. However, perfect recovery cannot be expected in reality due to noise in the quantum channel. This is why it becomes important to consider refinements of the DPI which would be physically meaningful, in the sense that all states $A$ for which the {\it relative entropy difference} $D(A|B)-D(\phi(A)|\phi(B))$ is small can be approximately recovered. The last decade saw a surge in activity on this topic, with influential work by Fawzi et. al. \cite{fr}, Sheshadreesan et. al. \cite{ws}, Wilde \cite{mw}, Junge et. al. \cite{jw}, Carlen et. al. \cite{cv}, Vershynina \cite{ver} and Gao et. al. \cite{gw}, some of which is discussed below.
	\medskip
	
	\textbf{Note :} Henceforth throughout the paper, the Petz recovery map will always be denoted by $\mathcal{R}$. The channel $\phi$ and density matrix $B$ will be clear from the context.
	\medskip
	
	Let $A, B\in M_n(\mathbb{C})$ be density matrices. Recall that the {\it fidelity} between $A$ and $B$ is defined as $$F(A|B)=\textrm{tr}(A^{1/2}BA^{1/2})^{1/2}.$$
	This measures how close the states $A$ and $B$ are. It is known that $F(A|B)$ is symmetric in $A$ and $B$, $0\leq F(A|B)\leq 1$ and that it is jointly concave, a proof of which can be found in \cite{uhl}. Let $\phi:M_n(\mathbb{C})\to M_k(\mathbb{C})$ be a channel. Consider the Stinespring representation of $\phi$ given by $$\phi(X)=\textrm{tr}_{2} VXV^{*}, X\in M_n(\mathbb{C})$$ where $V: \mathbb{C}^n\to\mathbb{C}^m\otimes\mathbb{C}^k$ is an isometry for some $m\in\mathbb{N}$. The Petz recovery map $\mathcal{R}$ defined in \eqref{e3} can then be written as \[\begin{aligned}\mathcal{R}(Y)&=B^{1/2}V^{*}(I_m\otimes\phi(B)^{-1/2}Y\phi(B)^{-1/2})VB^{1/2}\end{aligned}\] for all $Y\in M_{k}(\mathbb{C})$. Seshadreesan et. al. conjectured in \cite{ws} that \[-2\ln \big[F(A|\mathcal{R}(\phi(A)))\big]\leq D(A|B)-D(\phi(A)|\phi(B)).\label{e4}\tag{4}\] This arose from their study of a R\'enyi generalization of the relative entropy difference, which is defined in \cite{ws} and \cite{mw} as $$\tilde{\Delta}_t(A, B,\phi)=\frac{2t}{t-1}\ln\big|\big|(I_m\otimes\phi(A)^{\frac{1-t}{2t}}\phi(B)^{\frac{t-1}{2t}})VB^{\frac{1-t}{2t}}A^{\frac{1}{2}}\big|\big|_{2t}$$ for $t\in[\frac{1}{2},1)\cup(1,\infty)$.
	
	Seshadreesan et. al. (\cite{ws}) proved that $\tilde{\Delta}_t(A, B,\phi)$ converges to $D(A|B)-D(\phi(A)|\phi(B))$ as $t\to 1$ and conjectured that $\tilde{\Delta}_t(A, B,\phi)$ is monotonically increasing in $t$, thereby giving inequality \eqref{e4} as a natural consequence. This conjecture inspired a lot of subsequent research, leading to other significant refinements of the data processing inequality. Li and Winter discussed it along with some other related conjectures in \cite{lw}, and noted that it is favoured by numerical simulations. They proved the classical information theoretic version of the inequality as well. However, they also disproved a stronger inequality in the quantum setup with a counterexample. Wilde discussed the conjecture in \cite{mw}, and used interpolation techniques to prove that \[-2\ln \big[\sup_{t\in\mathbb{R}}F(A|\mathcal{R}^{t}_{\phi, B}(\phi(A)))\big]\leq D(A|B)-D(\phi(A)|\phi(B))\label{e5}\tag{5}\] where $\mathcal{R}^{t}_{\phi, B}$ is a {\it rotated} Petz recovery map given by \[\mathcal{R}^{t}_{\phi, B}(Y)=B^{\frac{1}{2}+it}\big(\phi(B)^{-\frac{1}{2}-it}Y\phi(B)^{-\frac{1}{2}-it}\big)B^{\frac{1}{2}+it}\label{e6}\tag{6}\] for each $t\in\mathbb{R}$.
	This coincides with the usual Petz recovery map $\mathcal{R}$ at $t=0$.
	Sutter et. al. (\cite{st}, \cite{bt}) and Junge et. al. (\cite{jw}) then took a significant leap to show the existence of recovery maps $\mathcal{R}^{\prime}$ depending only on $\phi$ and $B$ such that \[-2\ln \big[F(A|\mathcal{R}^{\prime}(\phi(A)))\big]\leq D(A|B)-D(\phi(A)|\phi(B)).\] This has recently been generalized to infinite dimensions in \cite{fh} and \cite{fh2}. 
	\medskip
	
	Amidst all the progress surrounding it, the conjecture itself still remained unsolved. We note that an approach to resolve it can be made by asking whether the map \[t\to\big|\big|(I_m\otimes\phi(A)^{\frac{1-t}{2t}}\phi(B)^{\frac{t-1}{2t}})VB^{\frac{1-t}{2t}}A^{\frac{1}{2}}\big|\big|_{2t}^{2t}\label{e7}\tag{7}\] is log-convex on $[\frac{1}{2},1]$. If true, this would imply the monotonicity of $\tilde{\Delta}_t(A, B,\phi)$ and inequality \eqref{e4} as its consequence. Similar log-convexity results in this paper helped in deducing monotonicity of some R\'enyi generalizations of the relative entropy. In section $3$ of this paper, we carry out a detailed analysis with $2\times 2$ matrices to demonstrate why the map in \eqref{e7} might fail to be log-convex, and then give an explicit counterexample to disprove inequality \eqref{e4}. 
	\medskip
	
	It now becomes natural to ask, can we have some approximate recoverability theorems with the Petz map itself? An important progress happened in 2020. Let $\mathcal{A}$ be a $C^{*}$-subalgebra of $M_n(\mathbb{C})$. Given $X\in M_n(\mathbb{C})$ let $X_0$ be its orthogonal projection onto $\mathcal{A}$ with respect to the Hilbert-Schmidt inner product, satisfying \[\mathrm{tr}\,XC=\mathrm{tr}\, X_0 C\label{e8}\tag{8}\] for all $C\in\mathcal{A}$. The map $X\to X_0$ is often called the conditional expectation onto $\mathcal{A}$ (see \cite{cl}). It was proved by Carlen and Vershynina \cite{cv} in 2020 that given a pair of positive definite density matrices $A$ and $B$, \[\begin{aligned}4\big[1-F(A|\mathcal{R}(A_0))\big]^4&\leq ||A-\mathcal{R}(A_0)||^4_1\\&\leq K||A^{-1}||^2[D(A|B)-D(A_0|B_0)].\end{aligned}\] where $K$ is a constant independent of $A$ and $\mathcal{R}$ is the Petz map corresponding to $\phi$ and $B$. Later works by Vershynina \cite{ver}, Gao and Wilde \cite{gw} generalized this to $f$-divergences and optimized $f$-divergences. This also includes the sandwiched quasi-relative entropies, which are defined for $t\in (0,1)\cup(1,\infty)$ by \[\mathcal{S}_t(A|B)=\mathrm{tr}(B^{\frac{1-t}{2t}}AB^{\frac{1-t}{2t}})^t.\]

	However, in all these results, the entropy difference is scaled by a factor which is an unbounded function of $A$. For example, in Carlen and Vershynina's result, the factor $||A^{-1}||^2$ appears, which cannot be bounded from above by a constant. To be precise, we are looking for an inequality to capture the phenomenon of {\it uniform} approximate recoverability. Qualitatively, this can be phrased as : 
	
	{\it Let $\mathcal{D}$ be a relative entropy for which Petz's theorem holds and let $B\in M_n(\mathbb{C})$ be positive definite. Let $\phi:M_n(\mathbb{C})\to M_k(\mathbb{C})$ be a channel. Then for every $\varepsilon\textgreater 0$ there exists a $\delta\textgreater0$ such that $F(A|\mathcal{R}\circ\phi(A))\geq 1-\varepsilon$ for all $A$ such that $\mathcal{D}(A|B)-\mathcal{D}(\phi(A)|\phi(B))\leq\delta.$, where $\mathcal{R}$ is the Petz recovery map.}
	
	The first such inequality is due to Cree and Sorce \cite{cs}, who showed that \[\begin{aligned}4\big[1-F(A|\mathcal{R}\circ\phi(A))\big]^2&\leq ||A-\mathcal{R}\circ\phi(A)||_1^2\\&\leq ||B||^2_2||B^{-1}||[\mathcal{S}_2(A|B)-\mathcal{S}_2(\phi(A)|\phi(B))].\end{aligned}\] This was recently refined by Gao et. al. \cite{lg} using information geometry techniques to show that \[4[1-F(A|\mathcal{R}\circ\phi(A))]^2\leq||A-\mathcal{R}\circ\phi(A)||_1\leq [\mathcal{S}_2(A|B)-\mathcal{S}_2(\phi(A)|\phi(B))].\]  We give a simple, alternate proof of the above and obtain the following approximate recoverability result with the sandwiched R\'enyi relative entropy $\mathcal{D}_2=\ln\hspace{1mm} \mathcal{S}_2$ : \[||A-\mathcal{R}(A_0)||_1^2\leq\mathcal{L}(||A||,||A_0||)||B^{-1}||[\mathcal{D}_2(A|B)-\mathcal{D}_2(A_0|B_0)].\label{e9}\tag{9}\] Here, $X\to X_0$ is a conditional expectation given by \eqref{e8} and $\mathcal{L}(a,b)$ denotes the log-mean between $a$ and $b$ given by $$L(a,b)=\frac{a-b}{\ln a-\ln b}.$$ Physically, inequality $\eqref{e9}$ tells us that the recoverability bound gets better for smaller values of $||A||$, implying that impure states are better recovered. We give an example to show inequality $\eqref{e9}$ is sharp as well. Additionally, our techniques yield the inequality :\[|||\mathrm{tr}_2 A|||.|||(\mathrm{tr}_2 B)^{-1}|||\leq\frac{|||A|||.|||B^{-1}|||}{|||I_n|||}\] for every unitarily invariant tensor norm $|||.|||$. In particular, this includes all Schatten$-p$ norms.
	\medskip

	In sections $2$ and $3$, we study some parametrized quantum entropies and discuss log-convexity with respect to the parameter. It is to be noted here that convexity problems of various kinds have always been of general interest in quantum information theory (see eg. \cite{gl, uhl, lf, ml, pt, bjl, sb, fl, rbh}).
	Given a parameter $\theta \in [0,1]$, the $\theta$-{\it divergence} of $A$ with respect to $B$ is defined to be $\textrm{tr}(A^{\theta} B^{1-\theta})$. Lieb's famous concavity theorem (\cite{el, ta, rbh}) asserts that this is jointly concave in $A$ and $B$. {\it The R\'enyi $\theta$-relative entropy} of $A$ with respect to $B$ is defined as \[D_{\theta}(A|B) = \frac{\ln\textrm{tr}(A^{\theta}B^{1 - \theta})}{\theta - 1}.\]
	This has been discussed in \cite{gl2}, \cite{pt} and \cite{ml}. It is known that $D_{\theta}(A|B)$ is monotonically increasing in $\theta$ and \[\mathop\lim\limits_{\theta\to 1-}D_{\theta}(A|B)= D(A|B).\] In section $2$, we strengthen this monotonicity result and prove that the function $\theta \to \textrm{tr}(A^{\theta}B^{1-\theta})$ is log-convex. We do this by proving a more general log-convexity result on $C^{*}$-algebras. We also give a refinement of Jensen's inequality for the function $t\to t\ln t$.
	\medskip

	We look at the sandwiched quasi-relative entropies $\mathcal{S}_t$ in section $3$. The corresponding sandwiched R\'enyi relative entropies are defined as \[\mathcal{D}_t(A|B)=\frac{\ln \mathcal{S}_t(A|B)}{t-1}.\] They have been introduced and studied in \cite{ml} and \cite{ww}. Beigi proved their DPI for $t\textgreater 1$ in \cite{sb}. Here we show that $t\to\mathcal{S}_t(A|B)$ is log-convex on $[\frac{1}{2}, 1]$ using interpolation techniques, giving us an alternate proof of the monotonicity of $\mathcal{D}_t(A|B)$ in $t$. We also use our results in section $3$ to give an alternate proof of the fact that \[\mathop\lim\limits_{t\to 1}\mathcal{D}_t(A|B)= D(A|B).\] It is to be noted here that the log-convexity results in sections $2$ and $3$ are mainly to motivate a similar approach towards inequality $\eqref{e4}$ in section $3$, where we carry out a thorough analysis along with an explicit $2\times 2$ counterexample to demonstrate why it fails.
	
	\medskip
	
	\section{The $\theta$-divergence}
	
	We start with a general log-convexity result that will help us prove log-convexity of the function $\theta\to\mathrm{tr}\hspace{1mm}A^{\theta}B^{1-\theta}$.
	\begin{thm}\label{t2} Let $\mathcal{A}$ be a unital $C^{*}$-algebra. Let $x\in\mathcal{A}$ be positive and invertible. Then for every positive definite state $\phi:\mathcal{A}\to\mathbb{C}$ the function $\theta\to\phi(x^{\theta})$ is log-convex on $[0,1]$.\end{thm}
	\begin{proof} 
		It suffices to show that for all $t,s \in [0,1]$, $$\ln\phi(x^{\frac{t+s}{2}})\leq\frac{\ln\phi(x^t)+\ln\phi(x^s)}{2}.$$
		Note that the $2\times2$ operator matrix $$\begin{pmatrix}x^t & x^{\frac{t+s}{2}}\\x^{\frac{t+s}{2}} & x^s\end{pmatrix}$$ is positive because the Schur complement of $x^t$ is $0$. Since $\phi$ is a state, it is completely positive and therefore, $$\begin{pmatrix}\phi(x^t) & \phi(x^{\frac{t+s}{2}})\\\phi(x^{\frac{t+s}{2}}) & \phi(x^s)\end{pmatrix} \geq 0.$$ Taking determinant, $$\phi(x^t)\phi(x^s)\geq\phi(x^{\frac{t+s}{2}})^{2}.$$
		Taking log on both sides, $$\ln\phi(x^{\frac{t+s}{2}})\leq\frac{\ln\phi(x^t)+\ln\phi(x^s)}{2}.$$ This completes the proof.
	\end{proof}
	
	Theorem \ref{t2} yields the following consequences :
	\begin{cor}\label{c3} Let $\mathcal{A}$ be a unital $C^{*}$-algebra and let $x\in\mathcal{A}$ be positive and invertible. Then for any positive definite state $\phi:\mathcal{A}\to\mathbb{C}$, $\phi(x\ln x)\geq \phi(x)\big[\ln\phi(x)+2(\ln \phi(x)^{1/2}-\ln\phi(x^{1/2}))\big]$.\end{cor}
	\begin{proof} Replacing $x$ with $\frac{x}{\phi(x)}$ if necessary, we may assume, without loss of generality, that $\phi(x)=1$. The inequality then reduces to $$\phi(x\ln x)\geq -2\ln\phi(x^{1/2}).$$
		By theorem \ref{t2}, $f(\theta)=\ln\phi(x^{\theta})$ is convex on $[0,1]$ and therefore, the function $h(\theta)=\frac{\ln\phi(x^{\theta})}{\theta-1}$ is increasing on $[0,1)$ and converges to $f^{\prime}(1)=\phi(x\ln x)$ as $\theta\to 1-$. This implies that $$h(1/2)=-2\ln\phi(x^{1/2})\leq\phi(x\ln x).$$\end{proof}
	\begin{cor}\label{c4}Let $A,B$ be positive definite density matrices. Then the function $\theta\to\ln\textrm{tr}(A^{\theta}B^{1-\theta})$ is convex on $[0,1]$,\end{cor}
	\begin{proof} Consider $M_n(\mathbb{C})$ as a Hilbert space with the Hilbert-Schmidt inner product $$\langle X,Y\rangle = \textrm{tr}(Y^{*}X).$$ Let $A,B$ be positive definite density matrices and $B(M_{n}(\mathbb{C}))$ be the $C^{*}$-algebra of operators on $M_{n}(\mathbb{C})$. Consider the pure state $\phi:B(M_n(\mathbb{C}))\to\mathbb{C}$ given by $$\phi(T)=\langle TB^{1/2},B^{1/2}\rangle$$ for all $T\in B(M_n(\mathbb{C}))$ and the relative modular operator $\Delta:M_n(\mathbb{C})\to M_n(\mathbb{C})$ given by $$\Delta(X)=AXB^{-1}$$ for all $X\in M_n(\mathbb{C})$. Being a composition of two commuting positive definite operators, namely the left multiplication by $A$ and the right multiplication by $B^{-1}$, $\Delta$ is also positive definite. We observe that $$\textrm{tr}(A^{\theta}B^{1-\theta})=\langle\Delta^{\theta}B^{1/2}, B^{1/2}\rangle = \phi(\Delta^{\theta})$$ and conclude the proof by invoking Theorem \ref{t2}.\end{proof}

	\begin{rem}
	\end{rem} 
	\begin{enumerate}
		\item The inequality in Corollary \ref{c3} is stronger than Jensen's inequality for the convex function $t\to t\ln t$ since $\phi(x)^{1/2}\geq \phi(x^{1/2})$.\\ 
		
		\item The relative modular operator used in Corollary \ref{c4} is positive and therefore, admits a functional calculus, which can be used to define a large class of quantum entropies called $f$-divergences. We consider a function $f:(0,\infty)\to\mathbb{R}$, often assumed to be operator convex or concave, and define the $f$-divergence of two positive definite density matrices $A$ and $B$ as $\langle f(\Delta)B^{1/2}, B^{1/2}\rangle$. This coincides with the Kullback-Liebler and the $\theta$-divergences when $f(x)=x\ln x$ and $x^{\theta}$ respectively. See \cite{hp} and \cite{dp} for more details on $f$-divergences.\\
		
		\item From Corollary \ref{c4}, we get the inequality \[-2\ln\textrm{tr}(A^{1/2}B^{1/2})\leq D(A|B).\] See \cite{lp} for a different proof. Since $\textrm{tr}(A^{1/2}B^{1/2})\leq F(A|B)$, it follows that $$-2\ln F(A|B)\leq D(A|B).$$ This has an important physical interpretation : if the divergence of $A$ with respect to $B$ is small, the fidelity between $A$ and $B$ is large.
	\end{enumerate}

	\section{Sandwiched R\'enyi entropy, relative entropy difference and recoverability}
	
	\begin{defn}\label{d1} Given density matrices $A$ and $B$ with $\textrm{supp}(A)\subset\textrm{supp}(B)$, define
	\[\mathcal{S}_{t}(A|B)=\textrm{tr}(B^{\frac{1-t}{2t}}AB^{\frac{1-t}{2t}})^{t}\] for all $t\in(0,\infty)$.\end{defn}  This is known as the sandwiched quasi-relative entropy, and was defined by M\"uller-Lennert et. al. (\cite{ml}) and Wilde et. al. (\cite{ww}) as a $t$-parametrized version of $F(A|B)$ Note that this coincides with the usual fidelity at $t=\frac{1}{2}$. Observe that $$\textrm{tr}(B^{\frac{1-t}{2t}}AB^{\frac{1-t}{2t}})^{t}\textgreater 0$$ for all $t\in(0, \infty)$ if $\textrm{supp}(A)\subset\textrm{supp}(B)$. It has been proved in \cite{lf} that $\mathcal{S}_{t}(A|B)$ is jointly concave in $A$ and $B$ if $t\in[\frac{1}{2}, 1)$ and jointly convex if $t\in(1,\infty)$. Some extremal characterizations of $\mathcal{S}_{t}(A|B)$ can be found in \cite{bjl}. For $t\in(0,1)\cup(1,\infty)$
	the sandwiched R\'enyi relative entropy is defined as \[\mathcal{D}_{t}(A|B)=\frac{\ln \mathcal{S}_{t}(A|B)}{t-1}.\label{e10}\tag{10}\] 
	
	The next theorem gives the log-convexity of $\mathcal{S}_{t}(A|B)$ in $[\frac{1}{2},1]$. For that, we need a matrix version of Riesz-Thorin interpolation. 
	\begin{lem}\label{l2} Let $S=\{z\in\mathbb{C}:0\leq\textrm{Re}(z)\leq1\}$ and $\Phi:S\to M_n(\mathbb{C})$ be a bounded continuous function that is holomorphic in the interior of $S$. Let $p_0, p_1\geq 1$. For $\theta\in[0,1]$ let $p_{\theta}$ given by $$\frac{1}{p_{\theta}}=\frac{1-\theta}{p_0}+\frac{\theta}{p_1}.$$
		If $$\sup_{t\in\mathbb{R}}||\Phi(\theta+it)||_{p_{\theta}}\textgreater 0$$ for all $\theta\in[0,1]$, the map $$\theta\to\sup_{t\in\mathbb{R}}||\Phi(\theta+it)||_{p_{\theta}}$$ is log-convex on $[0,1]$\end{lem}
	Here, for a given matrix $X$ and $r\geq 1$, $||X||_p$ denotes the Schatten-$p$ norm of $X$ given by $$||X||_p = (\textrm{tr}|X|^{p})^{1/p}.$$
	A proof of Lemma \ref{l2} can be found in \cite{sb} where interpolation theory is used to show the joint convexity of $\mathcal{F}_{t}(A|B)$ when $t\textgreater 1$. A detailed account of complex interpolation theory can be found in \cite{rs} and \cite{ss}.
	\begin{thm}\label{t7} The function $t\to\ln \mathcal{S}_{t}(A|B)$ is convex on $[\frac{1}{2},1]$\end{thm}
	\begin{proof} Let $t\in[\frac{1}{2},1]$. Then \[\begin{aligned}\mathcal{S}_t(A|B)&=\textrm{tr}(B^{\frac{1-t}{2t}}AB^{\frac{1-t}{2t}})^{t}\\
			&= \textrm{tr}|A^{\frac{1}{2}}B^{\frac{1-t}{2t}}|^{2t}\\&= \big|\big|A^{\frac{1}{2}}B^{\frac{1-t}{2t}}\big|\big|_{2t}^{2t}.\end{aligned}\]
		Let $h:[\frac{1}{2},1]\to\mathbb{R}$ be given by $h(t)=\ln \mathcal{S}_{t}(A|B)$. Define maps $\xi:[\frac{1}{2}, 1]\to[0,1]$ and $g:[0,1]\to\mathbb{R}$ by $$\xi(t)=\frac{1-t}{t}$$ and $$g(\theta)=\ln\big|\big|A^{\frac{1}{2}}B^{\frac{\theta}{2}}\big|\big|_{\frac{2}{1+\theta}}$$ respectively, and note that $$h(t)=\ln \mathcal{S}_{t}(A|B)=2t\hspace{0.75mm}g(\xi(t))=2t\hspace{0.75mm}g(\frac{1-t}{t}).$$ We now show that the convexity of $h$ on $[\frac{1}{2},1]$ is equivalent to the convexity of $g$ on $[0,1]$. Assume first that $g$ is convex. We have to show that for $t,s\in[\frac{1}{2},1]$, $$h\big(\frac{t+s}{2}\big)\leq\frac{h(t)+h(s)}{2}$$ which is equivalent to
		$$g\big(\frac{2}{t+s}-1\big)\leq\frac{t}{t+s}g\big(\frac{1-t}{t}\big)+\frac{t}{t+s}g\big(\frac{1-s}{s}\big).$$ But this is the same as
		$$g\big(\frac{t}{t+s}(\frac{1-t}{t})+\frac{s}{t+s}(\frac{1-s}{s})\big)\leq\frac{t}{t+s}g\big(\frac{1-t}{t}\big)+\frac{t}{t+s}g\big(\frac{1-s}{s}\big)$$
		which follows directly from the convexity of $g$. Conversely, assume $h$ is convex and observe that $$g(\theta)=\frac{(1+\theta)h\big(\frac{1}{1+\theta}\big)}{2}$$ for all $\theta\in [0,1]$. A similar argument shows that $g$ is convex. 
		
		Consider the strip $S=\{z\in\mathbb{C}:0\leq\textrm{Re}(z)\leq 1\}$ and the map $\Phi:S\to M_n(\mathbb{C})$ given by $$\Phi(z)=A^{\frac{1}{2}}B^{\frac{z}{2}}.$$
		It is easy to observe that $\Phi$ is continuous and bounded on $S$, and holomorphic in its interior. Consider the family of Schatten norms $$p_{\theta}=\frac{2}{1+\theta}$$ for $\theta\in[0,1]$. Then, for all $y\in\mathbb{R}$, $\theta\in[0,1]$,
		\[\begin{aligned}\big|\big|\Phi(\theta+iy)\big|\big|_{p_{\theta}}&=\big|\big|A^{\frac{1}{2}}B^{\frac{\theta}{2}}B^{\frac{iy}{2}}\big|\big|_{p_{\theta}}\\ &=\big|\big|A^{\frac{1}{2}}B^{\frac{\theta}{2}}\big|\big|_{p_{\theta}}\\&=\big|\big|\Phi(\theta)\big|\big|_{p_{\theta}}\end{aligned}\] Since $$g(\theta)=\ln\big|\big|\Phi(\theta)\big|\big|_{p_{\theta}}$$ for all $\theta\in[0,1]$, Lemma \ref{l2} implies the convexity of $g$, thereby proving the convexity of $h$.
	\end{proof}
	An immediate corollary follows :
	\begin{cor}\label{c5} The sandwiched R\'enyi relative entropy $\mathcal{D}_{t}(A|B)$ given by equation \eqref{e10} is monotonically increasing in $t$ on the interval $[\frac{1}{2},1)$\end{cor}
	
	Another proof of monotonicity can be found in \cite{ml}.

	We now look at a R\'enyi generalization of the relative entropy difference. Let $A, B\in M_n(\mathbb{C})$ be density matrices with $\textrm{supp}(A)\subset\textrm{supp}(B)$ and let $\phi:M_n(\mathbb{C})\to M_k(\mathbb{C})$ be a CPTP map. Consider the Petz recovery map $\mathcal{R}:M_{k}(C)\to M_{n}(\mathbb{C})$ given by $$\mathcal{R}(Y)=B^{1/2}\phi^{*}(\phi(B)^{-1/2}Y\phi(B)^{-1/2})B^{1/2}$$ for all $Y\in M_{k}(\mathbb{C})$. Observe that $\mathcal{R}(\phi(B))=B$. It has been proved (\cite{hp, pr}) that $R_{\phi, B}(\phi(A))=A$ if and only if $$D(A|B)=D(\phi(A)|\phi(B))$$
	By Stinespring's theorem, we have a positive integer $m$ and an isometry $V:\mathbb{C}^n\to\mathbb{C}^{m}\otimes\mathbb{C}^{k}$ such that $$\phi(X)=\textrm{tr}_{2}\hspace{0.75mm}VXV^{*}$$ for all $X\in M_n(\mathbb{C})$ where $\textrm{tr}_2$ denotes the second partial trace. Following \cite{mw} and \cite{ws} a R\'enyi generalization of the relative entropy difference can be defined as
	$$\tilde{\Delta}_t(A, B,\phi)=\frac{2t}{t-1}\ln\big|\big|(I_m\otimes\phi(A)^{\frac{1-t}{2t}}\phi(B)^{\frac{t-1}{2t}})VB^{\frac{1-t}{2t}}A^{\frac{1}{2}}\big|\big|_{2t}$$ for $t\in[\frac{1}{2},1)\cup(1,\infty)$. Note that \[\begin{aligned}&\big|\big|(I_m\otimes\phi(A)^{\frac{1-t}{2t}}\phi(B)^{\frac{t-1}{2t}})VB^{\frac{1-t}{2t}}A^{\frac{1}{2}}\big|\big|_{2t}^{2t}\\
		&=\textrm{tr}\big(A^{\frac{1}{2}}B^{\frac{1-t}{2t}}\phi^{*}(\phi(B)^{\frac{t-1}{2t}}\phi(A)^{\frac{1-t}{t}}\phi(B)^{\frac{t-1}{2t}})B^{\frac{1-t}{2t}}A^{\frac{1}{2}}\big)^{t}\end{aligned}\label{e12}\tag{12}\] for all $t\in[\frac{1}{2},1)$. At $t=\frac{1}{2}$ this coincides with $$F(A|\mathcal{R}(\phi(A)))=\textrm{tr}(A^{1/2}\mathcal{R}(\phi(A))A^{1/2})^{1/2}$$ which is the fidelity between $A$ and $\mathcal{R}(\phi(A))$. Hence, the expression in $t$ in \eqref{e12} can be seen as a parametrized version of this fidelity. It was proved in \cite{ws} by Seshadreesan et. al. and in \cite{mw} by Wilde that \[\mathop\lim\limits_{t\to 1}\tilde{\Delta}_t(A, B,\phi)= D(A|B)-D(\phi(A)|\phi(B)).\] In view of the previous results obtained, it is natural to wonder if the map $$t\to2t\hspace{0.75mm}\ln\big|\big|(I_m\otimes\phi(A)^{\frac{1-t}{2t}}\phi(B)^{\frac{t-1}{2t}})VB^{\frac{1-t}{2t}}A^{\frac{1}{2}}\big|\big|_{2t}$$ is convex on $[\frac{1}{2},1]$, which, by our argument in the proof of Theorem \ref{t7}, is equivalent to log-convexity of the map \[\theta\to\big|\big|(I_m\otimes\phi(A)^{\frac{\theta}{2}}\phi(B)^{-\frac{\theta}{2}})VB^{\frac{\theta}{2}}A^{\frac{1}{2}}\big|\big|_{\frac{2}{1+\theta}}\label{e13}\tag{13}\] on $[0,1]$. If this is true, both the monotonicity of $\tilde{\Delta}_t(A,B,\phi)$ and inequality \eqref{e4}, as conjectured by Seshadreesan et. al. in \cite{ws}, would follow. It is to be noted that this is what happens in classical information theory, as demonstrated in \cite{lw}. So our first step would be to observe the behaviour of this map for some special cases. The following theorem sheds some light on its log-convexity for certain choices of $A$.
	
	\begin{thm}\label{t9} Let $A\in M_{n}(\mathbb{C})$ be a pure state given by $A=x\otimes x^*$ and let $B\in M_{n}(\mathbb{C})$ be positive definite. Let $P$ be the projection orthogonal to $x$. Let $\phi:M_n(\mathbb{C})\to M_n(\mathbb{C})$ be the spectral pinching along $A$, defined by $\phi(C)=\langle Cx,x\rangle x\otimes x^*+PCP$ for all $C\in M_n(\mathbb{C})$. Then \[-2\ln[F(A|\mathcal{R}(\phi(A)))]\leq D(A|B)-D(\phi(A)|\phi(B)).\].\end{thm}
	\begin{proof} By \eqref{e13}, the inequality follows if the function $g:[0,1]\to\mathbb{R}$ given by $$g(\theta)=||A\phi(B)^{-\theta/2}B^{\theta/2}A||_{\frac{2}{1+\theta}}$$ is log-convex. Let $$B=\sum_j\sigma_j v_j\otimes v_j^{*}$$ be the spectral decomposition of $B$ and $\lambda=\langle Bx,x\rangle$. Observe that $$g(\theta)=\lambda^{-\theta/2}\langle B^{\theta/2}x,x\rangle$$ $$=\sum_j \big(\frac{\sigma_j}{\lambda}\big)^{\theta/2}|\langle x, v_j\rangle|^2$$ which is log-convex by Theorem \ref{t2}.
	\end{proof}
	
	Henceforth, we assume $A\in M_2(\mathbb{C})$ is a pure state and $\phi: M_2(\mathbb{C})\to M_2(\mathbb{C})$ is the pinching along the main diagonal. We could have taken the pinching along any orthonormal basis, but that can simply be reduced to the main diagonal pinching by changing co-ordinates. If $A$ is a candidate for a counterexample, Theorem \ref{t9} suggests that the off-diagonal entries of $A$ must be non-zero. Let \[x=\begin{pmatrix} s\\ t\end{pmatrix}\] where $s, t\in\mathbb{C}\setminus\{0\}$ with $|s|^2+|t|^2=1$. 
	Let \[A=x\otimes x^*= \begin{pmatrix}|s|^2 & s\bar{t}\\ \bar{s}t & |t|^2\end{pmatrix}.\] Let \[B=\begin{pmatrix}a & \bar{c}\\ c & b\end{pmatrix}\] be a positive definite density matrix. Since $\phi$ is the pinching along the diagonal, it is the conditional expectation onto the diagonal algebra $$\mathcal{A}=\big\{\begin{pmatrix} a_1 & 0\\ 0 & a_2\end{pmatrix}: a_1,\,a_2\in\mathbb{C}\big\}.$$ The map in \eqref{e13} now becomes \[g(\theta)= ||\phi(A)^{\theta/2}\phi(B)^{-\theta/2}B^{\theta/2}A||_{\frac{2}{1+\theta}}\label{e14}\tag{14}.\] for all $\theta\in [0,1]$. It should be noted that this is the case for any conditional expectation. Consider the spectral decomposition \[B=\lambda\hspace{0.7mm} v_1\otimes v_1^* +\mu\hspace{0.7mm} v_2\otimes v_2^*.\] A straightforward calculation shows that \[\begin{aligned}g(\theta)^2 &= \alpha^{2\theta}(\lambda^{\theta}|\gamma\langle v_1,e_1\rangle|^2+\mu^{\theta}|\delta \langle v_2,e_1\rangle|^2)\\&+\beta^{2\theta}(\lambda^{\theta}|\gamma\langle v_1,e_2\rangle|^2+\mu^{\theta}|\delta \langle v_2,e_2\rangle|^2)\\&+2\lambda^{\theta/2}\mu^{\theta/2}(\alpha^{2\theta}-\beta^{2\theta})\textrm{Re}(\gamma\bar{\delta} \langle v_1,e_1\rangle \overline{\langle v_2, e_1\rangle})\end{aligned}\] for all $\theta\in [0,1]$. Here $\alpha=|s|a$, $\beta=|t|b$, $z=\langle x, v_1\rangle$, and $w=\langle x, v_2\rangle$. Again, the sum of the first two terms is log-convex, so problems with log-convexity can only arise if the third term is non-zero. We can, infact, further simplify things by taking $x$ and $B$ to have real entries, in which case the third term will simply become \[2\lambda^{\theta/2}\mu^{\theta/2}\gamma\delta\langle v_1,e_1\rangle \langle v_2, e_1\rangle(\alpha^{2\theta}-\beta^{2\theta}).\label{e15}\tag{15}\] The spectral theorem for real symmetric matrices ensures that $v_1$ and $v_2$ have real entries, so $\gamma$ and $\delta$ are real. Thus, $g(\theta)^2$ (and hence, $g(\theta)$) is log-convex if the expression in \eqref{e15} vanishes, thereby giving the desired inequality. Hence, any candidate for a counterexample should not make this term vanish. To further simplify, lets take \[A=\begin{pmatrix}\frac{1}{2} & \frac{1}{2}\\[6pt]\frac{1}{2} & \frac{1}{2}\end{pmatrix}\] so that it is only required to choose $B$ judiciously. Now, since $s=t=\frac{1}{\sqrt{2}}$, the diagonal entries $a$ and $b$ of $B$ cannot be equal, as it would then imply $\alpha^{2\theta}-\beta^{2\theta}=0$. Also if $B$ is diagonal, $\langle v_1, e_1\rangle \langle v_2, e_1\rangle=0$. Hence, we need to choose $B$ such that it is not diagonal, and the diagonal entries are not equal. We tried out several examples, and the following disproved inequality \eqref{e4}, thus disproving the conjecture of Seshadreesan et. al. as well as log-convexity of the map in \eqref{e13} :
	
	\begin{example}\label{ex1} Let $$A=\begin{pmatrix}\frac{1}{2} & \frac{1}{2}\\[6pt]\frac{1}{2} & \frac{1}{2}\end{pmatrix}$$ and $$B=\begin{pmatrix}\frac{3}{4} & -\frac{1}{4}\\[6pt]-\frac{1}{4} & \frac{1}{4}\end{pmatrix}$$ be two density matrices. $A$ is a rank one projection and $B$ is positive definite. Let $\phi:M_2(\mathbb{C})\to M_2(\mathbb{C})$ be the pinching along the main diagonal, so that $$\phi(A)=\begin{pmatrix}\frac{1}{2} & 0\\[6pt]0 & \frac{1}{2}\end{pmatrix}$$ and $$\phi(B)=\begin{pmatrix}\frac{3}{4} & 0\\[6pt]0 & \frac{1}{4}\end{pmatrix}.$$ Computations using GNU Octave then reveal that $$D(A|B)-D(\phi(A)|\phi(B))\approx 1.5191$$ while $$-2\ln F(A|\mathcal{R}(\phi(A)))\approx 1.5349.$$
	\end{example}
	
	The failure of inequality $\eqref{e4}$ is essentially due to non-commutativity, which leads to the expression in $\eqref{e15}$ becoming non-zero. This destroys log-convexity of $g(\theta)$ (see $\eqref{e14}$) and causes inequality $\eqref{e4}$ to fail, even for simple $2\times 2$ examples where $A$ is a pure state. Let $A=x\otimes x^*\in M_2(\mathbb{C})$. Let $y\in \mathbb{C}^2$ be a unit vector such that $y\perp x$. As the above example demonstrates, if $\phi$ is not the pinching along the orthonormal basis $\{x,y\}$, enough non-commutativity is introduced in the system so that inequality \eqref{e4} can fail. Contrast this with the classical case in \cite{lw} and Theorem \ref{t9}. This demonstrates that the quantum case is significantly more challenging, as the conjectured inequality capturing approximate recoverability with respect to the K-L divergence fails even for the simplest of examples. However, from our analysis on log-convexity, we can deduce inequality \eqref{e4} for a particular class of pure states $A$, when $\phi$ is the pinching along any orthonormal basis.
	\medskip
	
	\begin{thm}Let $B\in M_2(\mathbb{C})$ be an invertible density matrix. Let $\{v,w\}$ be an orthonormal basis of $\mathbb{C}^2$ and let $\phi:M_2(\mathbb{C})\to M_2(\mathbb{C})$ be the pinching along $\{v,w\}$ given by $\phi(X)=\langle X v,v\rangle v\otimes v^* + \langle Xw, w\rangle w\otimes w^*$. Consider the set of all pure states $\mathbb{C}P^{1}=\{x\otimes x^* : x\in\mathbb{C}^2\,||x||=1\}$. Then there exists a simple closed loop $\tilde{A}:[0,2\pi]\to \mathbb{C}P^1$ such that \[-2\ln[F(\tilde{A}(\theta)|\mathcal{R}\circ \phi(\tilde{A}(\theta)))]\leq D(\tilde{A}(\theta)|B)-D(\phi(\tilde{A}(\theta))|\phi(B))\] for all $\theta\in [0,2\pi]$.\end{thm}
	
	\begin{proof}Without loss of generality, we can take $\phi$ to be the pinching along the main diagonal. Let $$B=\begin{pmatrix}a & \bar{c}\\ c & b\end{pmatrix}$$ and let \[A= \begin{pmatrix}|s|^2 & s\bar{t}\\ \bar{s}t & |t|^2\end{pmatrix}\] be a pure state. The map $g$ in \eqref{e14} is log-convex if the expression in \eqref{e15} vanishes, which happens if $\alpha=\beta$, or $|s|a=|t|b$. This forces $$A=\begin{pmatrix}\frac{b^2}{a^2+b^2} & \frac{z\sqrt{ab}}{a^2+b^2}\\\frac{\bar{z}\sqrt{ab}}{a^2+b^2} & \frac{a^2}{a^2+b^2}\end{pmatrix}$$ for some $|z|=1$. Let $\tilde{A}:[0,2\pi]\to\mathbb{C}P^1$ be given by $$\tilde{A}(\theta)=\begin{pmatrix}\frac{b^2}{a^2+b^2} & e^{i\theta}\frac{\sqrt{ab}}{a^2+b^2}\\ e^{-i\theta}\frac{\sqrt{ab}}{a^2+b^2} & \frac{a^2}{a^2+b^2}\end{pmatrix}.$$ This is the required loop in $\mathbb{C}P^1$. \end{proof}
	\medskip
	
	\textbf{Further questions :} $\mathbb{C}P^1$ is diffeomorphic to the sphere $S^2$ in $\mathbb{R}^3$, which is why it is often called the {\it Bloch sphere} in quantum information language (see \cite{wilde}). The above theorem demonstrates, with the use of log-convexity, that for any pinching in $M_2(\mathbb{C})$, we have a one-dimensional submanifold of the Bloch sphere where inequality \eqref{e4} holds. Though as example \ref{ex1} demonstrates, this need not be true for {\it all} qubits in the Bloch sphere. Therefore we may ask, given the main diagonal pinching $\phi$ and a $2\times 2$ positive definite density matrix $B$, how big is the set of all pure states $A$ such that inequality \eqref{e4} holds? It is clearly a closed subset of positive measure, as it contains a neighborhood. Can we have a good enough lower bound on its measure as a subset of $S^2$? Judging by the complexity of our expressions in the non-commutative setup, this seems to be quite non-trivial. We hope some future research can shed light on this. 
	\medskip
	
	\begin{rem}\end{rem}
	
	\begin{enumerate} 
		
		\item It is to be noted that the function in \eqref{e13} is log-convex for all $A$ such that $\phi(A)=\phi(B)$, and hence, inequality $\eqref{e4}$ holds. However, this is an extreme case, since this condition implies $D(\phi(A)|\phi(B))=0$. A physical interpretation of this is that the channel $\phi$ destroys all distinguishability information of $A$ relative to $B$.
		\medskip
		
\item Even though inequality \eqref{e4} fails to be true with the Petz recovery map, interpolation techniques have been used by Wilde in \cite{mw} and Junge et. al. in \cite{jw} to refine the data processing inequality and show the existence of recovery maps for which the analogous versions of inequality \eqref{e4} hold. In particular, it has been shown in \cite{jw} that there exists a recovery map $\mathcal{R}^{\prime}_{\phi, B}$ depending only on $\phi$ and $B$ such that \[-2\ln[F(A|\mathcal{R}^{\prime}_{\phi, B}(\phi(A)))]\leq D(A|B)-D(\phi(A)|\phi(B))\label{e16}\tag{16}.\] The map $\mathcal{R}^{\prime}$ is given by $$\mathcal{R}^{\prime}_{\phi, B}(Y)=\int_{\mathbb{R}}\mathcal{R}^{t}_{\phi, B}(Y)d\mu(t)$$ for a Borel probability measure $\mu$ on $\mathbb{R}$, where $\mathcal{R}^{t}_{\phi,B}$ is the rotated Petz recovery map defined in equation \eqref{e6}.\end{enumerate}
	
	\section{Approximate recovery with the Petz map}
	
	In view of the counterexample in the previous section, it is natural to wonder whether we can have inequalities which can be interpreted in terms of approximate recoverability with the Petz map. Carlen and Vershynina \cite{cv} obtained the first result of this kind. They showed that given a $C^{*}$-subalgebra $\mathcal{A}$ of $M_n(\mathbb{C})$ and positive definite density matrices $A$ and $B$, \[\begin{aligned}&[1-F(A|\mathcal{R}(A_0))]^4\\&\leq ||A-\mathcal{R}(A_0)||^4_1\\&\leq K||A^{-1}||^2[D(A|B)-D(A_0|B_0)]\end{aligned}\label{e17}\tag{17}\] where $K$ is a constant independent of $A$, and $A_0, B_0$ are the conditional expectations of $A$ and $B$ onto $\mathcal{A}$. Vershynina \cite{ver} generalized inequality \eqref{e17} to quantum $f$-divergences, and later, Gao and Wilde \cite{gw} proved similar results for optimized $f$-divergences, including the sandwiched quasi-relative entropies. However, in all these results, the relative entropy difference is scaled by factors which are dependent on $A$ in such a way that that it is not possible to bound them from above by a constant. This begs the question, when can such dependency be removed? Inequalities achieving this would then capture the phenomenon of {\it uniform} approximate recoverability with the Petz map. This is the main focus of this section. We prove a sharp universal bound with the Petz map for the sandwiched quasi-relative entropy $\mathcal{S}_2$, by utilizing the underlying Araki-Masuda Hilbert space geometry and showing how the Petz map fits naturally within this context.
	
	First we show that uniform approximate recoverability is qualitatively possible for continuous divergences.
	\begin{thm}  Let $\mathcal{D}$ be a jointly continuous divergence on density matrices. Let $B\in M_n(\mathbb{C})$ be a positive definite density matrix and let $\phi:M_n(\mathbb{C})\to M_k(\mathbb{C})$ be a channel. Then the following are equivalent :
		\begin{enumerate}
			\item[(i)] For any density matrix $A$, $\mathcal{D}(A|B)=\mathcal{D}(\phi(A)|\phi(B))$ if and only if $\mathcal{R}\circ\phi(A)=A$.
			
			\item[(ii)] for every $\varepsilon\textgreater 0$ there exists $\delta\textgreater 0$ such that $F(A|\mathcal{R}\circ\phi(A))\geq 1-\varepsilon$ for all $A$ with $\mathcal{D}(A|B)-\mathcal{D}(\phi(A)|\phi(B))\leq\delta$.
		\end{enumerate}
	\end{thm} 
	\begin{proof} (ii) implies (i) is obvious. To show the converse, we argue by contradiction. Let there exist $\varepsilon\textgreater 0$ and a sequence $\{A_n\}$ of density matrices such that $\mathcal{D}(A_n|B)-\mathcal{D}(\phi(A_n)|\phi(B))\to 0$ but $F(A_n|\mathcal{R}\circ\phi(A_n))\leq 1-\varepsilon$. Choosing a subsequence $\{A_{n_j}\}$ converging to some density matrix $A_1$, and using continuity of $\mathcal{D}$ and $F$, \[\mathcal{D}(A_1|B)=\mathcal{D}(\phi(A_1)|\phi(B))\] and \[F(A_1|\mathcal{R}\circ\phi(A_1))\leq 1-\varepsilon.\] But from (i), $\mathcal{R}\circ\phi(A_1)=A_1$, implying $F(A_1|\mathcal{R}\circ\phi(A_1))=1$ which is a contradiction .\end{proof}
	The next theorem refines inequality \eqref{e17} to show that the factor $||A^{-1}||^2$ can be replaced with $||A^{-1}||$. We can achieve this by modifying the proof in Carlen and Vershynina \cite{cv} to a certain degree.
	\begin{thm}\label{t10}Let $A, B\in M_n(\mathbb{C})$ be positive definite density matrices and let $\mathcal{A}\subset M_n(\mathbb{C})$ be a unital $C^{*}$-subalgebra. Then there exists a constant $K$ depending only on $B$ such that \[||A-\mathcal{R}(A_0)||^4_1\leq K||A^{-1}||[D(A|B)-D(A_0|B_0)]\] where $A_0, B_0$ denote the conditional expectations of $A$ and $B$ respectively onto $\mathcal{A}$, given by \eqref{e8}. \end{thm}
	\begin{proof}Let $\Delta$ and $\Delta_0$ be the relative modular operators corresponding to $A, B$ and $A_0, B_0$ respectively. Let $V:\mathcal{A}\to M_n(\mathbb{C})$ be the isometry given by $V(X)=XB_0^{-1/2}B^{1/2}$ for all $X\in\mathcal{A}$. Note that $V^{*}\Delta V=\Delta_0$. The Umegaki relative entropy is the $f$-divergence corresponding to the function $x\to x\ln x$ on $(0,\infty)$. Using the integral representation \[x\ln x=\int_{0}^{\infty}\big(\frac{x}{1+t}-\frac{x}{x+t}\big) dt\] we get\[\begin{aligned}&D(A|B)-D(A_0|B_0)\\&=\int_{0}^{\infty}\big[\langle\Delta_0(\Delta_0+t)^{-1}B_0^{1/2}, B_0^{1/2}\rangle - \langle \Delta(\Delta+t)^{-1}B^{1/2}, B^{1/2}\rangle\big]dt\\&=\int_{0}^{\infty}t\big[\langle (\Delta+t)^{-1}B^{1/2}, B^{1/2}\rangle-\langle (\Delta_0+t)^{-1}B_0^{1/2}, B_0^{1/2}\rangle\big]dt\\&=\int_{0}^{\infty}t\langle (\Delta+t)w_t, w_t\rangle dt\end{aligned}\] where \[w_t= V(\Delta_0+t)^{-1}B_0^{1/2}-(\Delta+t)^{-1}B^{1/2}\] for all $t\in(0,\infty)$. The smallest eigenvalue of $\Delta$ is $\frac{1}{||A^{-1}|| ||B||}$, which implies \[D(A|B)-D(A_0|B_0)\geq \frac{1}{||A^{-1}|| ||B||}\int_{0}^{\infty}t||w_t||_2^2 dt.\label{e18}\tag{18}\] The rest of the argument is the same as that in Carlen and Vershynina \cite{cv}. We include it here for the sake of completeness. Using the integral representation \[X^{1/2}=\frac{1}{\pi}\int_{0}^{\infty}t^{1/2}\big(\frac{1}{t}-\frac{1}{X+t}\big)dt\] for the square root, we obtain \[A^{1/2}-A_0^{1/2}B_0^{-1/2}B^{1/2}=\frac{1}{\pi}\int_{0}^{\infty}t^{1/2}w_t dt.\label{e19}\tag{19}.\] Taking the Frobenius norm on both sides, \[||A^{1/2}-A_0^{1/2}B_0^{-1/2}B^{1/2}||_2\leq\frac{1}{\pi}[\int_{0}^{\delta}t^{1/2}||w_t||_2 dt+||\int_{\delta}^{\infty}t^{1/2}w_t dt||_2]\] for any $\delta\textgreater 0$. Applying Cauchy-Schwarz on the first term, and inequality \eqref{e18}, \[\int_{0}^{\delta}t^{1/2}||w_t||_2 dt\leq\big[\delta||A^{-1}|| ||B|| \big(D(A|B)-D(A_0|B_0)\big)\big]^{1/2}.\label{e20}\tag{20}\] To estimate the second term, we observe that \[\begin{aligned}&\big|\big|\int_{\delta}^{\infty}t^{1/2}w_t dt\big|\big|_2\\&=\big|\big|\int_{\delta}^{\infty}t^{1/2}[V(\Delta_0+t)^{-1}B_0^{1/2}-(\Delta+t)^{-1}B^{1/2}]dt\big|\big|_2\\&=\big|\big|\int_{\delta}^{\infty}[V t^{1/2}(t^{-1}-(\Delta_0+t)^{-1})B_0^{1/2}-t^{1/2}(t^{-1}-(\Delta+t)^{-1})B^{1/2}]dt\big|\big|_2\\&=\big|\big|\int_{\delta}^{\infty}[V \frac{\Delta_0(\Delta_0+t)^{-1}}{t^{1/2}}B_0^{1/2}-\frac{\Delta(\Delta+t)^{-1}}{t^{1/2}}B^{1/2}]dt\big|\big|_2.\end{aligned}\label{e21}\tag{21}\] Note that for any positive operator $X$ on a Hilbert space, \[||X(t+X)^{-1}||=\frac{||X||}{t+||X||}\] for any $t\textgreater 0$. Hence, for any $\delta\textgreater0$, \[\begin{aligned}&\int_{\delta}^{\infty}\big|\big|\frac{X(t+X)}{t^{1/2}}\big|\big|dt\\&\leq\int_{\delta}^{\infty}\frac{||X||}{t^{1/2}(t+||X||)}dt\\&\leq\frac{2||X||}{\delta^{1/2}}.\end{aligned}\] Using this and \eqref{e21}, \[\begin{aligned}&\big|\big|\int_{\delta}^{\infty}t^{1/2}w_t dt\big|\big|_2\\&\leq\frac{2}{\delta^{1/2}}\big(||\Delta||+||\Delta_0||\big)\\&\leq\frac{2}{\delta^{1/2}}\big(||A||||B^{-1}||+||A_0||||B_0^{-1}||\big).\end{aligned}\] Since the conditional expectation is an operator norm contraction, $||A_0||\leq||A||$. The function $x\to x^{-1}$ is operator convex on $(0,\infty)$, and hence, Jensen's inequality implies \[(B_0)^{-1}\leq (B^{-1})_0.\] Combining these, we have \[\begin{aligned}&\big|\big|\int_{\delta}^{\infty}t^{1/2}w_t dt\big|\big|_2\\&\leq\frac{4}{\delta^{1/2}}||A||||B^{-1}||\\&\leq\frac{4}{\delta^{1/2}}||B^{-1}||.\end{aligned}\label{e22}\tag{22}\] By \eqref{e21} and \eqref{e22}, \[\begin{aligned}&||A^{1/2}-A_0^{1/2}B_0^{-1/2}B^{1/2}||_2\\&\leq\frac{1}{\pi} \big[\delta||A^{-1}|| ||B|| \big(D(A|B)-D(A_0|B_0)\big)\big]^{1/2}+\frac{4}{\delta^{1/2}}||B^{-1}||.\end{aligned}\]
		Minimizing over $\delta$ and using Lemma $2.2$ of \cite{cv}, \[\begin{aligned}&||A-\mathcal{R}(A_0)||_1\\&\leq 2||A^{1/2}-A_0^{1/2}B_0^{-1/2}B^{1/2}||_2\\&\leq\frac{8}{\pi}||B^{-1}||^{1/2}||A^{-1}||^{1/4}[D(A|B)-D(A_0|B_0)].\end{aligned}\] Taking $4^{\textrm{th}}$ powers and denoting $\frac{4096||B^{-1}||^2}{\pi^4}$ by $K$, \[\begin{aligned}&||A-\mathcal{R}(A_0)||_1^4\\&\leq K||A^{-1}||[D(A|B)-D(A_0|B_0)].\end{aligned}\]   \end{proof}
		\medskip
		
		\begin{rem2}Judging by the careful estimates used in the proof, we expect the factor $||A^{-1}||$ in Theorem \ref{t10} to be sharp, in the sense that \[\frac{||A-\mathcal{R}(A_0)||_1^4}{[D(A|B)-D(A_0|B_0)]}\] grows like $||A^{-1}||$ as $||A^{-1}||\to\infty$. If true, this would imply that a uniform approximate recoverability bound of the form \[\begin{aligned}&||A-\mathcal{R}(A_0)||_1^4\\&\leq C[D(A|B)-D(A_0|B_0)].\end{aligned}\] for some fixed $C\,\textgreater\, 0$ cannot be obtained. In light of example \ref{ex1} and the proof of Theorem \ref{t10}, we believe this to be the case. However, we were not able to demonstrate this behaviour with explicit examples. This remains a question to be addressed in future research. \end{rem2}
		\medskip
		
	Let $B$ be a positive definite density matrix. It was proved by Jen\v{c}ov\'a in \cite{jen} that given a channel $\phi:M_n(\mathbb{C})\to M_k(\mathbb{C})$ and density matrices $A, B$ with supp$(A)\subset$ supp$(B)$, for all $p\geq 1$, \[\mathcal{S}_p(\phi(A)|\phi(B))=\mathcal{S}_p(A|B)\hspace{1.7mm}\textrm{iff}\hspace{1.7mm}\mathcal{R}(\phi(A))=A\] where $\mathcal{R}$ is the Petz recovery map and $\mathcal{S}_p$ is the sandwiched quasi-relative entropy given by definition \ref{d1}.
	Let $B\in M_n(\mathbb{C})$ be a positive definite density matrix. Consider the inner product on $M_n(\mathbb{C})$ given by $$\langle X, Y\rangle_B = \textrm {tr} (B^{-1/2}Y^*B^{-1/2}X)$$ for all $X, Y\in M_{n}(\mathbb{C})$. From definition \ref{d1} it is clear that \[||A||^2_B=\langle A, A\rangle^{1/2}=\mathcal{S}_2(A|B).\] Let $\phi:M_n(\mathbb{C})\to M_k(\mathbb{C})$ be a channel and assume $\phi(B)$ is positive definite too (otherwise, we can restrict the target Hilbert space to be supp$(\phi(B))$). Equipping $M_k(\mathbb{C})$ with the inner product \[\langle Z, W\rangle_{\phi(B)}=\textrm{tr}(\phi(B)^{-1/2}W^*\phi(B)^{-1/2}Z)\] for all $Z, W\in M_k(\mathbb{C})$, we observe that the adjoint of $\phi$ as a map from $(M_n(\mathbb{C}), \langle \rule{2mm}{0.1mm},\rule{2mm}{0.1mm}\rangle_{B})$ to $(M_k(\mathbb{C}), \langle \rule{2mm}{0.1mm},\rule{2mm}{0.1mm}\rangle_{\phi(B)})$ is the Petz recovery map $\mathcal{R}$. As shown by Jen\v{c}ov\'a in \cite{jen}, this gives a simple proof of recoverability with the Petz map for the sandwiched quasi-relative entropy $S_2$. In the subsequent theorems, we use some elementary Hilbert space techniques to derive refinements of the DPI for $S_2$ meaningful in the context of approximate recoverability. We state and prove some lemmas which will be useful in obtaining the results.
	\begin{lem}\label{l3} Let $\mathcal{H}, \mathcal{K}$ be Hilbert spaces and let $T:\mathcal{H}\to\mathcal{K}$ be a contraction. Then we have the following : \begin{enumerate}
			\item[(i)] for all $x\in \mathcal{H}$, \[||x-T^*Tx||^2\leq ||x||^2-||Tx||^2,\hspace{1mm}\mathrm{and}\]
			\item[(ii)] for all $x$ such that $||x||\textgreater 1$, \[||Tx||^2\geq\frac{||Tx||^2-1}{||x||^2-1}||x||^2.\]
	\end{enumerate}\end{lem}
	\begin{proof} For part (i), observe that \[||x-T^*Tx||^2=||x||^2-2||Tx||^2+\langle (T^*T)^2x,x\rangle.\] Since $(T^*T)^2\leq T^*T$, it follows that \[\begin{aligned}||x-T^*Tx||^2&\leq||x||^2-2||Tx||^2+\langle T^*Tx,x\rangle\\ &= ||x||^2-||Tx||^2\end{aligned}.\] 
		
		Part (ii) is a simple consequence of the fact that $||Tx||^2\leq ||x||^2$ for all $x\in\mathcal{H}$.
	\end{proof}
	\begin{lem}\label{l4} Let $B\in M_n(\mathbb{C})$ be a positive definite density matrix. Then \begin{enumerate}\item[(i)]for all density matrices $A$, \[||A||^2_B\leq ||B^{-1}|| ||A||,\hspace{1mm}\mathrm{and}\]
			\item[(ii)]for all $X\in M_n(\mathbb{C})$, \[||X||^2_1\leq ||X||^2_B,\] where $||.||_1$ denotes the trace norm of a matrix.\end{enumerate}\end{lem}
	\begin{proof} Let $\sigma_1\geq\sigma_{2}\geq\cdots\geq\sigma_n\textgreater 0$ be the eigenvalues of $B$ with corresponding orthonormal eigenvectors $\{v_j\}_{j=1}^n$. For part (i), note that the Araki-Lieb-Thirring inequality (see IX.2.10 of \cite{rb}) implies \[\begin{aligned}||A||^2_B&=\textrm{tr}(B^{-1/4}AB^{-1/4})^2\\&\leq\textrm{tr}A^2B^{-1}\\&=\sum_j\frac{1}{\sigma_j}\langle A^2v_j, v_j\rangle\\&\leq\frac{1}{\sigma_n}\textrm{tr}A^2\\&\leq||B^{-1}||||A||\textrm{tr}A\\&=||B^{-1}||||A||\end{aligned}.\]
		For part (ii), assume without loss of generality that $B$ is diagonal. Let \[B=\begin{pmatrix}
			\sigma_{1} & & \\
			& \ddots & \\
			& & \sigma_{n}
		\end{pmatrix}\] where $\sigma_1\geq\sigma_{2}\geq\cdots\geq\sigma_n\textgreater 0$. Let $X=(x_{ij})$. Then, \[\begin{aligned}||X||_B^2&=\sum_{i,j}\frac{|x_{ij}|^2}{\sqrt{\sigma_i\sigma_j}}\\&\geq\sum_{i,j}\frac{2|x_{ij}|^2}{\sigma_i+\sigma_j}.\end{aligned}\]
		Let $U=[u_{ij}]$ be a unitary. Then \[\begin{aligned}|\mathrm{tr\hspace{1mm}}U^*X|^2&=|\sum_{i,j}x_{ij}\bar{u}_{ij}|^2\\&=|\sum_{i,j}\sqrt{\frac{2}{\sigma_i+\sigma_j}}x_{ij}\sqrt{\frac{\sigma_i+\sigma_j}{2}}\bar{u}_{ij}|^2\\&\leq\big(\sum_{i,j}\frac{2|x_{ij}|^2}{\sigma_i+\sigma_j}\big)\frac{\big(\sum_{i}\sigma_i\sum_{j}|u_{ij}|^2+\sum_{j}\sigma_j\sum_{i}|u_{ij}|^2\big)}{2}\\&=\sum_{i,j}\frac{2|x_{ij}|^2}{\sigma_i+\sigma_j}\\&\leq ||X||_B^2.\end{aligned}\] Maximizing over unitaries, we are done.
	\end{proof} 
	
	Let $A, B$ be $n\times n$ density matrices where $B$ is positive definite and $A\neq B$. Let $\phi:M_n(\mathbb{C})\to M_k(\mathbb{C})$ be a channel and assume $\phi(B)$ is positive definite. As observed before, $\phi$ is a contraction from $(M_n(\mathbb{C}),\langle\rule{1.7mm}{0.1mm},\rule{1.7mm}{0.1mm}\rangle_B)$ to $(M_k(\mathbb{C}),\langle\rule{1.7mm}{0.1mm},\rangle_{\phi(B)}))$. Let $\mathcal{W}_{A,B}$ be the subspace spanned by $A$ and $B$. Restricting $\phi$ to $\mathcal{W}_{A,B}$, we see that its adjoint is $P_{A,B}\circ\mathcal{R}$, where $P_{A,B}$ is the orthogonal projection onto $\mathcal{W}_{A,B}$. Observe that $P_{A,B}\circ\mathcal{R}\circ\phi(B)=B$ and $\langle A-B,B\rangle_{B}=0$. Since $\phi$ is a contraction, the spectral theorem applied to $P_{A,B}\circ\mathcal{R}\circ\phi$ now implies the existence of a $\lambda_{A,B}\in[0,1]$ such that \[P_{A,B}\circ\mathcal{R}\circ\phi(A-B)=\lambda_{A,B}(A-B)\] which implies \[P_{A,B}\circ\mathcal{R}\circ\phi(A)=\lambda_{A,B}A+(1-\lambda_{A,B})B.\] Its easy to see that $\lambda_{A,B}=1$ iff $\mathcal{R}(\phi(A))=A$. Infact, we have an explicit expression for $\lambda_{A,B}$, given by \[\begin{aligned}\lambda_{A,B}&=\frac{\langle\mathcal{R}\circ\phi(A),A\rangle_B-1}{||A||^2_B-1}\\&=\frac{||\phi(A)||^2_{\phi(B)}-1}{||A||^2_{B}-1}\\&=\frac{\mathcal{S}_2(\phi(A)|\phi(B))-1}{\mathcal{S}_2(A|B)-1}.\end{aligned}\label{e23}\tag{23}\] Intuitively, $\lambda^{\phi}_{A,B}$ is a measure of how close $\mathcal{R}(\phi(A))$ is to $A$, and this is made precise in the following theorem.
	\begin{thm}\label{t11} Let $A, B$ be $n\times n$ density matrices such that $B$ is positive definite and $A\neq B$. Let $\phi:M_n(\mathbb{C})\to M_k(\mathbb{C})$ be a channel. Then \[||A-\mathcal{R}(\phi(A))||^2_2\leq ||A-\mathcal{R}(\phi(A))||^2_{B}||B||\leq ||B||(||B^{-1}||-1)(1-\lambda_{A,B})\].\end{thm}
	\begin{proof} The first inequality is obvious from Lemma \ref{l4}. To deduce the second, note that by Lemma \ref{l3}, \[\begin{aligned}\hspace{4mm}||\mathcal{R}\circ\phi(A)-A||^2_B&\leq ||A||^2_B-||\phi(A)||^2_{\phi(B)}\\&=(||A||^2_B-1)(1-\lambda_{A,B})\hspace{5mm}[\textrm{by}    
			\hspace{2mm}\eqref{e23}]\\&\leq(1-\lambda_{A,B})(||B^{-1}||-1)\end{aligned}\] where the last inequality follows from Lemma \ref{l4}. \end{proof}
	Part (ii) of Lemma \ref{l3} also yields the following corollary :
	\begin{cor}\label{c6} Let $A, B$ be $n\times n$ density matrices such that $B$ is positive definite and $A\neq B$. Let $\phi:M_n(\mathbb{C})\to M_k(\mathbb{C})$ be a channel. Then \[\mathcal{S}_2(\phi(A)|\phi(B))\geq \lambda_{A,B}\mathcal{S}_2(A|B).\]\end{cor}
	Corollary \ref{c6} has an interesting physical interpretation. Recall that $\mathcal{S}_2$, as a relative entropy, is a measure of distinguishability and $\lambda_{A,B}$ is a measure of closeness between $A$ and $\mathcal{R}\circ\phi(A)$. The inequality then states that the channel $\phi$ preserves atleast as much distinguishability information as there is between $A$ and $B$ scaled by a factor depending on how accurately the Petz map $\mathcal{R}$ recovers $A$.
	We now prove a result connecting approximate recoverability with the Petz map and the sandwiched R\'enyi relative entropy $\mathcal{D}_2$.
	\begin{thm}\label{t12} Let $A, B\in M_n(\mathbb{C})$ be density matrices such that $B$ is positive definite and let $\phi: M_n(\mathbb{C})\to M_k(\mathbb{C})$ be a channel. Let $q:[1,\infty]\to[1,\infty]$ be the harmonic conjugate map given by $q(t)=\frac{t}{t-1}$ for all $t\in [1,\infty]$. Then \[-\ln \lambda_{A,B}\leq q(||\phi(A)||^2_{\phi(B)})[\mathcal{D}_2(A|B)-\mathcal{D}_2(\phi(A)|\phi(B))].\]\end{thm}
	\begin{proof} It suffices to prove the inequality for states $A$ such that $||\phi(A)||^2_{\phi(B)}\textgreater 1$. Denoting $q(||\phi(A)||^2_{\phi(B)})$ by $\rho$, we note that the inequality is equivalent to \[\frac{||A||^{2\rho}_{B}}{||A||^2_B-1}\geq\frac{||\phi(A)||^{2\rho}_{\phi(B)}}{||\phi(A)||^2_{\phi(B)}-1}.\label{e24}\tag{24}\]
		Consider  $f:(1,\infty)\to\mathbb{R}$ given by \[f(t)=\frac{t^{\rho}}{t-1}.\] It is easy to see that for all $t\geqslant ||\phi(A)||^2_{\phi(B)}$, $f^{\prime}(t)\geqslant 0$, thereby showing that $f$ is increasing on $\{t:t\geqslant||\phi(A)||^2_{\phi(B)}\}$. Inequality \eqref{e18} now follows since $||A||^2_{B}\geqslant||\phi(A)||^2_{\phi(B)}$.
	\end{proof}
	As a simple corollary, we have :
	\begin{cor}Let $A, B\in M_n(\mathbb{C})$ be density matrices such that $B$ is positive definite and let $\phi: M_n(\mathbb{C})\to M_k(\mathbb{C})$ be a channel. Let $q:[1,\infty]\to[1,\infty]$ be the harmonic conjugate map given by $q(t)=\frac{t}{t-1}$ for al $t\in [1,\infty]$. Then \[-\ln \big[1-\frac{||A-\mathcal{R}\circ\phi(A)||_2^2}{\kappa(B)-||B||}\big]\leq q(||\phi(A)||^2_{\phi(B)})[\mathcal{D}_2(A|B)-\mathcal{D}_2(\phi(A)|\phi(B))]\] where $\kappa(B)$ denotes the condition number of $B$.\end{cor}
	\begin{rem2} In Theorem \ref{t12}, the scale factor on the right and side of the inequality is dependent of $\phi(A)$, and infact, it tends towards infinity the closer $\phi(A)$ gets to $\phi(B)$. Thus, for any $\delta\textgreater 0$, the inequality \[-\ln \lambda_{A,B}\leq (1+\frac{1}{\delta})[\mathcal{D}_2(A|B)-\mathcal{D}_2(\phi(A)|\phi(B))]\] holds for all $A$ such that $||\phi(A)||^2_{\phi(B)}\geq 1+\delta$. Informally, this means that all $A$ for which $\phi(A)$ and $\phi(B)$ are distinguishable beyond a predetermined threshold can be approximately recovered with the Petz map.\end{rem2}
	We state and prove our next theorem, which refines the DPI for $\mathcal{S}_2$ and can be interpreted in terms of approximate recoverability with the fidelity. This was initially obtained by Gao et. al. very recently in \cite{lg} using quantum Fisher Information techniques. For details on the Fisher information, refer to \cite{pet1, pet}.
	
	\begin{thm}\label{t13} Let $A, B\in M_n(\mathbb{C})$ be density matrices such that $B$ is positive definite and let $\phi:M_n(\mathbb{C})\to M_k(\mathbb{C})$ be a channel. Then, \[4\big[1-F(A|\mathcal{R}\circ\phi(A))\big]^2\leq ||A-\mathcal{R}\circ\phi (A)||_1^2\leq [\mathcal{S}_2(A|B)-\mathcal{S}_2(\phi(A)|\phi(B))].\]\end{thm}
	\begin{proof}Observe that \[\begin{aligned}& 4[1-F(A|\mathcal{R}\circ\phi(A))]^2\\&\leq ||A-\mathcal{R}\circ\phi(A)||_1^2\\&\leq ||A-\mathcal{R}\circ\phi(A)||_B^2\end{aligned}\] where the last inequality follows from lemma \ref{l4}. Now, by lemma \ref{l3}, \[\begin{aligned}&||A-\mathcal{R}\circ\phi(A)||_B^2\\&\leq[\mathcal{S}_2(A|B)-\mathcal{S}_2(\phi(A)|\phi(B))].\end{aligned}\] \end{proof}
	\medskip
	
	The sharpness of this inequality is discussed later in example \ref{ex2}.
	\medskip
	
	As an immediate corollary, we deduce the following refinement of the DPI for $\mathcal{D}_2$ :
	\begin{cor}\label{c7}Let $A, B\in M_n(\mathbb{C})$ be density matrices such that $B$ is positive definite and let $\phi:M_n(\mathbb{C})\to M_k(\mathbb{C})$ be a channel. Then, \[||A-\mathcal{R}\circ\phi(A)||_1^2\leq ||A||||B^{-1}||[\mathcal{D}_2(A|B)-\mathcal{D}_2(\phi(A)|\phi(B))].\] \end{cor}
	\begin{proof} By the mean value theorem applied to the map $t\to\ln t$ on the interval $\{t:||\phi(A)||^2_{\phi(B)}\leq t\leq ||A||^2_B\}$, we see that \[\begin{aligned}&||A||^2_B-||\phi(A)||^2_{\phi(B)}\\&\leq||A||^2_B[\ln||A||^2_B-\ln||\phi(A)||^2_{\phi(B)}]\\&=||A||^2_B[\mathcal{D}_2(A|B)-\mathcal{D}_2(\phi(A)|\phi(B))]\\&\leq||A||||B^{-1}||[\mathcal{D}_2(A|B)-\mathcal{D}_2(\phi(A)|\phi(B))]\end{aligned}\] where the last  inequality follows from Lemma \ref{l4}.\end{proof}
	
	Corollary \ref{c7} gives an approximate recoverability result with the sandwiched R\'enyi relative entropy $\mathcal{D}_2$. The scaling factor of $||A||$ on the right hand side tells us that the bound works better the smaller $||A||$ is. Thus, it says that mixed states can be recovered to a higher degree of accuracy than pure states.
	\vspace{2mm}
	
	It is natural to ask whether the factor $||A||$ in corollary \ref{c7} can be further refined. This is explored in the next theorems. For a pair of positive reals $a$ and $b$, we denote their logarithmic mean by $\mathcal{L}(a,b)$.
	\begin{lem}\label{l5} Let $A, B\in M_n{C}$ be density matrices and let $\phi:M_n(\mathbb{C})\to M_k(\mathbb{C})$ be a channel. Then, \[||A-\mathcal{R}\circ\phi(A)||_1^2\leq\mathcal{L}(||A||.||B^{-1}||, ||\phi(A)||.||\phi(B)^{-1}||)[\mathcal{D}_2(A|B)-\mathcal{D}_2(\phi(A)|\phi(B))].\]\end{lem}
	\begin{proof}By theorem \ref{t13}, \[\begin{aligned}&||A-\mathcal{R}\circ\phi(A)||_1^2\\&\leq\mathcal{S}_2(A|B)-\mathcal{S}_2(\phi(A)|\phi(B))\\&=\mathcal{L}(||A||_B^2, ||\phi(A)||_{\phi(B)}^2)[\mathcal{D}_2(A|B)-\mathcal{D}_2(\phi(A)|\phi(B))]\\&\leq\mathcal{L}(||A||.||B^{-1}||, ||\phi(A)||.||\phi(B)^{-1}||)[\mathcal{D}_2(A|B)-\mathcal{D}_2(\phi(A)|\phi(B))].\end{aligned}\] The last inequality follows from lemma \ref{l4} and monotonicity of the log mean.\end{proof}
	The inequality in lemma \ref{l5} is a refinement of corollary \ref{c7} whenever $||\phi(A)||.||\phi(B)^{-1}||\leq||A||.||B^{-1}||$. The next theorem shows that this is always true for conditional expectations.
	\begin{thm}\label{t14} Let $\mathcal{A}$ be a $C^{*}$-subalgebra of $M_n(\mathbb{C})$. For any $X\in M_n(\mathbb{C})$ let $X_0$ be its conditional expectation onto $\mathcal{A}$. Then, for all $A, B\in M_n(\mathbb{C})$ such that $A$ is positive and $B$ is positive definite, we have, \[|||A_0|||.|||B_0^{-1}|||\leq|||A|||.|||B^{-1}|||\] for any unitarily invariant norm $|||.|||$ on $M_n(\mathbb{C})$.\end{thm}
	\begin{proof}Note that \[A_0=\sum_{j=1}^k\mu_jU_jAU_j^*\] for some $\{\mu_j\}_{j=1}^k$ such that $\mu_j\geq 0$ for all $j$, $\sum_j\lambda_j=1$ and unitaries $U_j\in\mathcal{A}^{\prime}$, where $\mathcal{A}^{\prime}$ denotes the commutant of $\mathcal{A}$. Thus, $$|||A_0|||\leq|||A|||$$ for any unitarily invariant norm. Also, by the operator Jensen inequality, $$B_0^{-1}\leq (B^{-1})_0.$$ By Weyl's inequalities, \[\lambda^{\downarrow}_{j}(B_0^{-1})\leq\lambda^{\downarrow}_{j}((B^{-1})_0)\] where $\lambda^{\downarrow}_{j}$ denotes the $j^{\mathrm{th}}$ largest eigenvalue. Hence, the eigenvalues of $B_0^{-1}$ are weakly majorized by the eigenvalues of $(B^{-1})_0$. Therefore, $$|||B^{-1}_0|||\leq|||(B^{-1})_0|||\leq|||B^{-1}|||$$ for any unitarily invariant norm $|||.|||$.\end{proof}
	Since $|||B_0^{-1}|||\leq|||B^{-1}|||$ for any unitarily invariant norm $|||.|||$, it follows that for any two density matrices $A, B$ and a conditional expectation $X\to X_0$ onto some $C^*$-subalgebra $\mathcal{A}\subset M_n(\mathbb{C})$, \[||A-\mathcal{R}(A_0)||_1^2\leq\mathcal{L}(||A||,||A_0||)||B^{-1}||[\mathcal{D}_2(A|B)-\mathcal{D}_2(A_0|B_0)].\label{e25}\tag{25}\] This improves the factor $||A||$ in corollary \ref{c7} to $\mathcal{L}(||A||, ||A_0||)$.
	
	We now discuss sharpness of the inequalities obtained.
	
	\begin{example}\label{ex2}Let $\mathcal{A}\subset M_{2n}(\mathbb{C})$ be the subalgebra generated by $I_{2n}$. Take $$A=\frac{1}{n}\begin{pmatrix}I_n & O\\O & O\end{pmatrix}$$ and $B=\frac{1}{2n}I_{2n}$. Let $\phi:M_{2n}(\mathbb{C})\to M_{2n}(\mathbb{C})$ be the conditional expectation onto $\mathcal{A}$.
		\medskip
		
		It is easy to see that under these assumptions, $||A-\mathcal{R}\circ\phi(A)||_1^2$ and $\mathcal{S}_2(A|B)-\mathcal{S}_2{\phi(A)|\phi(B)}$ are both $1$, thus giving sharpness of the inequality in theorem \ref{t13}. Also, the same example demonstrates sharpness of inequality \eqref{e25}. Note, however, that this example doesn't work for corollary \ref{c7} as in this case, $\mathcal{L}(||A||, ||A_0||)\textless ||A||.$ So we have indeed obtained a better bound with inequality \eqref{e25}.\end{example}
		\medskip
		
	Its not true in general that $\mathcal{L}(||A||.||B^{-1}||, ||\phi(A)||.||\phi(B)^{-1}||)\leq ||A||.||B^{-1}||$ for any channel $\phi$. For example, take any $t\in (1/2,1)$ and consider $\phi:M_2(\mathbb{C})\to M_2(\mathbb{C})$ given by \[\phi(X)=\mathrm{tr }X\begin{pmatrix}t & 0\\0 & 1-t\end{pmatrix}\] for all $X\in M_2(\mathbb{C})$. Take density matrices $A=B=I/2$. Then, $||A||.||B^{-1}||=1$ but \[\begin{aligned}&\mathcal{L}(||A||.||B^{-1}||, ||\phi(A)||.||\phi(B)^{-1}||)\\&=\mathcal{L}\big(1,\frac{t}{1-t}\big)\textgreater 1.\end{aligned}\] However, the refinement holds if $\phi$ is a partial trace, which the next inequality demonstrates.
	\begin{thm}\label{t15} Let $A, B\in M_n(\mathbb{C})\otimes M_k(\mathbb{C})$ be positive matrices such that $B$ is invertible. Let $|||.|||$ be any unitarily invariant tensor norm on matrices. Then \[|||\mathrm{tr}_2 A|||.|||(\mathrm{tr}_2 B)^{-1}|||\leq\frac{|||A|||.|||B^{-1}|||}{|||I_n|||}.\]\end{thm}
	\begin{proof}Note that \[\begin{aligned}&|||\mathrm{tr}_2 A|||.|||(\mathrm{tr}_2 B)^{-1}|||\\&=\frac{|||I_n\otimes\frac{\mathrm{tr}_2 A}{n}|||.|||(I_n\otimes\frac{\mathrm{tr}_2 B}{n})^{-1}|||}{||I_n|||^2}\\&\leq \frac{|||A|||.|||B^{-1}|||}{|||I_n|||^2}\end{aligned}\] where the last inequality follows from Theorem \ref{t14} because for all $X\in M_n(\mathbb{C})\otimes M_k(\mathbb{C})$, $I_n\otimes\frac{\mathrm{tr}_2 X}{n}$ is the conditional expectaion of $X$ onto the subalgebra $\{I_n\otimes Y:Y\in M_k(\mathbb{C})\}$.\end{proof}
	
	\begin{rem}
	\end{rem}
	\begin{enumerate}
		\item Theorem \ref{t15} holds for $\mathrm{tr}_1$ as well. In this case, the denominator on the right hand side becomes $|||I_k|||$ instead of $|||I_n|||$.
		\vspace{2mm}
		
		\item For the operator norm $||.||$ the inequality in theorem \ref{t15} reads \[||\mathrm{tr}_2 A||.||(\mathrm{tr}_2 B)^{-1}||\leq||A||.||B^{-1}||.\] Note that the inequality does not really hold component-wise, for instance, $||\mathrm{tr}_2 A||$ is bounded from above by $n||A||$ and not $||A||$.
	\end{enumerate}
	\medskip
	
	\section{Conclusions}
	In this section we present a brief summary of the work done in this paper. 
	\medskip
	
	\begin{enumerate}
	\item Inequality \eqref{e4}, as conjectured in \cite{ws}, relied on the monotonicity of a R\'enyi generalization of the relative entropy difference $D(A|B)-D(\phi(A)|\phi(B))$. This R\'enyi generalization is analogous to the standard R\'enyi relative entropies and the sandwiched R\'enyi entropies. In sections 2 and 3, we show how log-convexity can be used as a unifying theme to prove monotonicity of both these entropies with respect to the parameter $t$. Motivated by this, we try to replicate the same for the R\'enyi generalization of relative entropy difference $$\tilde{\Delta}_t(A, B,\phi)=\frac{2t}{t-1}\ln\big|\big|(I_m\otimes\phi(A)^{\frac{1-t}{2t}}\phi(B)^{\frac{t-1}{2t}})VB^{\frac{1-t}{2t}}A^{\frac{1}{2}}\big|\big|_{2t}$$ introduced in \cite{ws}. We carry out a detailed analysis in section 3 and demonstrate that in this case, non-commutativity becomes an unavoidable roadblock, causing both log-convexity and the conjectured inequality \eqref{e4} to fail even in the $2\times 2$ case, as demonstrated in example \ref{ex1}. This demonstrates how radically different quantum approximate recoverability is compared to its classical counterpart. We still obtain some partial results for a class of qubits, and leave some further questions open.
	\medskip
	
	\item  Section 4 is dedicated to results capturing approximate recoverability with the Petz map. In theorem \ref{t10}, we refine an inequality originally proved by Carlen and Vershynina \cite{cv}, thus resulting in a factor of $||A^{-1}||$ before the relative entropy difference $D(A|B)-D(\phi(A)|\phi(B))$. We conjecture that this factor is optimal, which would be consistent with our other findings confirming that approximate recoverability with respect to the K-L divergence is not well-behaved in the quantum setup.
	\medskip
	
	\item Further down section $4$, we go on to demonstrate how the sandwiched quasi-relative entropies $\mathcal{S}_p$ arise from a class of modified Schatten norms on $M_n(\mathbb{C})$. In particular, for $p=2$, the Petz map $\mathcal{R}$ arises naturally as the adjoint of the channel $\phi$ with respect to this modified Hilbert space structure. This seems to indicate that this is the most natural setting to carry out analysis on approximate recoverability with the Petz map. Our intuition is proven right, as we obtain numerous universal recoverability bounds utilizing this structure, notably, theorem \ref{t13} and inequality \eqref{e25}. Both these inequalities are sharp, as demonstrated in example \ref{ex2}.
	\medskip
	
	\item Inequality \eqref{e25} seems to be particularly important from the perspective of quantum information theory, as it demonstrates that mixed states can be recovered with better precision than pure states. We would like to see some further research expanding on this phenomenon.
	\medskip
	
	\item A very interesting perspective on the inequality in theorem \ref{t13} can be found in the recent paper by Gao et. al. \cite{lg}, where recoverability is discussed from the perspective of differential geometry on density matrices. We note, however, that our proof is much simpler, and the core structure of the proof survives in infinite dimensions, with enough necessary technical modifications. Indeed, in our recent paper \cite{sap1}, we use similar techniques, along with the required non-trivial adjustments, to generalize this inequality to von-Neumann algebras. This gives, to the best of our knowledge, the first universal recoverability bound in infinite dimensions, thereby solidifying our intuition that it is the sandwiched quasi-relative entropies that are better suited for recoverability with the Petz map.
	\medskip

	\end{enumerate}
	
	\begin{acknowledgement} I thank my PhD supervisor Prof. Tanvi Jain for her valuable comments and suggestions during the preparation of this paper. I also thank Prof. Rajendra Bhatia for his insightful comments on improving the exposition. Finally, I thank the anonymous reviewer for taking time to read the manuscript in detail and suggest several necessary improvements.\end{acknowledgement}

	\textbf{Conflict of interest :} I declare no conflict of interest.
	\medskip
	
	\textbf{Data availability :}  No available data has been used. 
	
	\bibliographystyle{amsplain}
		
\end{document}